\documentclass[11pt,a4paper]{article}
\usepackage{color,soul}
\usepackage{amsmath,amssymb,amsthm,latexsym}
\usepackage[english]{babel}
\usepackage{graphicx}
\usepackage{enumitem}
\usepackage{comment}
\usepackage{authblk}
\usepackage[autostyle]{csquotes}
\usepackage{circuitikz}
\usepackage{tikz,graphicx,subfigure}
\usetikzlibrary{hobby}
\usepackage{hyperref}
\usepackage[a4paper, total={6in, 9in}]{geometry}
\usepackage[numbers,sort]{natbib}

\usepackage{todonotes}
\hypersetup{
    colorlinks=true,
    linkcolor=black,
    filecolor=magenta,
    citecolor=black,
    urlcolor=cyan,
    pdftitle={DL25},
    pdfpagemode=FullScreen,
    }
\usetikzlibrary{decorations.markings}
\usetikzlibrary{arrows}

\newtheorem{theorem}{Theorem}
\usepackage{color,hyperref}
\hypersetup{
  colorlinks   = true, 
  urlcolor     = MyBlue1, 
  linkcolor    = MyBlue1, 
  citecolor   = darkspringgreen 
}
\definecolor{MyOra}{RGB}{255, 165, 0}
\definecolor{MyOra1}{RGB}{255,119,0}
\definecolor{MyBlue}{rgb}{0.25,0.5,0.75}

\definecolor{darkgreen}{rgb}{0.0,0.5,0.0}
\definecolor{darkspringgreen}{rgb}{0.05, 0.5, 0.06}
\usepackage[]{youngtab}

\usepackage{xcolor}
\definecolor{MyBlue1}{rgb}{0.0,0.0,0.55}
\colorlet{NextBlue}{MyBlue!20}
\colorlet{SecondBlue}{MyBlue!40}

\newtheorem{lemma}[theorem]{Lemma}
\newtheorem{proposition}[theorem]{Proposition}
\newtheorem{corollary}[theorem]{Corollary}

\theoremstyle{definition}
\newtheorem{definition}[theorem]{Definition}

\newtheorem{remark}[theorem]{Remark}

\allowdisplaybreaks[4]

\let\Im\undefined

\DeclareMathOperator{\Im}{Im}

\numberwithin{equation}{section}
\numberwithin{theorem}{section}

\def\XXint#1#2#3{{\setbox0=\hbox{$#1{#2#3}{\int}$}
		\vcenter{\hbox{$#2#3$}}\kern-.5\wd0}}

\newcounter{rhp}
\newenvironment{rhp}[1][]{\refstepcounter{rhp}\par\medskip
   \noindent \textbf{Riemann-Hilbert problem~\therhp. #1} \rmfamily}{\medskip}

\makeatletter

\date{}                     

\begin{document}

\title{Recurrence relations and the Christoffel-Darboux formula for elliptic orthogonal polynomials}

\author[1]{Harini Desiraju}
\author[2,3]{Sampad Lahiry}
\affil[1]{
Sydney mathematical Research Institute, University of Sydney, NSW 2008, Australia}
\affil[2]{
Department of Mathematics, Katholieke Universiteit Leuven, Celestijnenlaan
200 B bus 2400, 3001 Leuven,  Belgium}
\affil[3]{School of Mathematical and Statistics, The University of Melbourne, Victoria 3010, Australia}
\affil[  ]{harini.desiraju@sydney.edu.au, sampad.lahiry@kuleuven.be}
\maketitle

\begin{abstract}
In recent years, there has been significant progress in the theory of orthogonal polynomials on algebraic curves, particularly on genus 1 surfaces. In this paper, we focus on elliptic orthogonal polynomials and establish several of their fundamental properties. In particular, we derive general five-term and seven-term recurrence relations, which lead to a Christoffel–Darboux formula and the construction of an associated point process on the A-cycle of the torus. Notably, the recurrence coefficients in these relations are intricately linked through the underlying elliptic curve equation. Under additional symmetry assumptions on the weight function, the structure simplifies considerably, recovering known results for orthogonal polynomials on the complex plane.

\end{abstract}

\tableofcontents

\newpage
\section{Introduction}

The Christoffel-Darboux (CD) formula is central to the spectral theory of
orthogonal polynomials (OPs) (see \cite{simon2008christoffel} for a review), and is influential in several other areas such as integrable systems, random matrices, point processes, approximation theory \cite{nevai1986geza}, and data analysis \cite{lasserre2022christoffel}. The CD formula is a classical object of study for OPs on the complex plane. In the recent years, there have been important advances in the understanding of orthogonal polynomials and integrable systems on higher genus Riemann surfaces \cite{fasondini2023orthogonal,bertola2022nonlinear,DLR24,del2020isomonodromic,desiraju2025nonlinear}. Following these developments, an obvious question, due to the proximity between OPs and interacting particle systems, is if polynomials on higher genus surfaces could lead to novel phenomenon for point processes. In this paper, we achieve the first step towards this goal, namely constructing the CD  formula, for elliptic orthogonal polynomials (EOPs)  as constructed in \cite{DLR24} with degree independent weights.

Let us now introduce the EOPs we deal with in this paper. 


Consider the complex torus defined by the lattice with modular parameter \( \tau \in i\mathbb{R} \):
\begin{equation*}
    \mathbb{T} := \mathbb{C} / (\mathbb{Z} + \mathbb{Z} \cdot \tau).
\end{equation*}
An \emph{elliptic orthogonal polynomial} (EOP) \( P_k(z) \), $k\in \mathbb{Z}_{+}$ is defined as a doubly periodic meromorphic function on \( \mathbb{T} \) that has a unique pole of order \( k \) at $z=0$. Since no such function exists for $k=1$, we set $P_{1}(z)=0$ for the rest of our article. We also assume that \( P_k \) is \emph{monic} {\it i.e},  it admits the following Laurent expansion near the origin:
\begin{equation}
    P_k(z) = z^{-k}\left(1 + \mathcal{O}(z)\right), \quad \text{as } z \to 0.
\end{equation}
Such polynomials can be expressed in terms of a basis constructed from Weierstrass \( \wp \)-functions, given by
\begin{equation}
    \mathcal{B} = \{ \mathcal{E}_n \}_{n \geq 0,\, n \neq 1}, \qquad 
    \mathcal{E}_{2k} = \wp(z)^k, \qquad 
    \mathcal{E}_{2k+3} = -\tfrac{1}{2} \wp'(z)\, \wp(z)^k.
\end{equation}
Importantly, the polynomials \( \{ P_k \} \) satisfy the orthogonality relation
\begin{equation}
    \int_{\gamma}  P_k(z)\, P_m(z)\, {w}(z)\, dz = h_k\, \delta_{k,m}, \quad k, m \in \{0,2,3,\dots\},
\end{equation}
where \( h_k \) is a normalization constant, and \( \delta_{k,m} \) denotes the Kronecker delta.
The weight \( {w}(z) \) is assumed to be a positive function supported on the cycle
\[
    \gamma := \left[ \frac{\tau}{2}, 1 + \frac{\tau}{2} \right].
\]
Finally, analogous to the classic result on the complex plane \cite{fokas1992isomonodromy}, the polynomials $P_{k}$ also satisfies the following Riemann-Hilbert problem (RHP) \cite{DLR24,bertola2022nonlinear}.

\begin{rhp}\label{RHPY}
The RHP comprises of finding a $2\times 2$ matrix valued function $Y_n(z,\tau)$ with the following properties:
\begin{itemize}
    \item $Y_n(z,\tau)$ is analytic in $z\in \mathbb{T}\setminus ( \gamma \cup \{0\} )$.
    \item The following jump condition hold for $z\in \gamma$:
    \begin{align}\label{jump:Yn}
        Y_{n,+}(z,\tau) = Y_{n,-}(z,\tau) \left(\begin{array}{cc}
          1   &  w(z)\\
          0   & 1
        \end{array} \right),
    \end{align}
    where, following the standard notation,
    $\pm$ indicate the piece-wise analytic functions to the left and right side respectively of $\gamma$ w.r.t its orientation,
    \item In the limit $z\to 0$:
    \begin{align}
        Y_n(z,\tau)=\left( I_2 + \mathcal{O}(z) \right) \left( \begin{array}{cc}
         z^{-n}    & 0 \\
           0  & z^{n-2}
        \end{array}\right).
    \end{align}
\end{itemize}
\end{rhp}
Indeed, RHP \ref{RHPY} is uniquely solvable and has the form
\begin{equation}\label{RHPsol}
    Y_{n}(z,\tau)=\begin{pmatrix}
        P_n(z)& \mathcal{C}( P_n)(z)\\
        \frac{2\pi i}{h_{n-1}}  P_{n-1}(z)& \frac{2\pi i}{h_{n-1}}\mathcal{C}( P_{n-1})(z)
    \end{pmatrix} \quad n\geq 3,
\end{equation}
where $\mathcal{C}(.)$ is the Cauchy transform of a function $f\in L^1(\gamma)$, defined as
\begin{align*}
     \mathcal{C}(f)(z):= \int_{\gamma} f(w)(\zeta(w-z)-\zeta(w))\frac{dw}{2\pi i},
\end{align*}
 where $\zeta(.)$ denotes the Weierstrass $\zeta$-function.

For the reminder of the paper, we deal with orthonormal polynomials denoted by $\pi_n$. For convenience, let us establish the following definition.
\begin{definition}\label{def:orthonormalEOPs}
 The family of EOPs \( \{\pi_n(z)\}_{n \geq 0,\, n \neq 1} \) are orthonormal with respect to the weight function \( w(z) \) on the contour \( \gamma \).
\end{definition}
Indeed, they are obtained by the rescaling
\begin{equation}
    \pi_{n}=\frac{1}{\sqrt{h_n}}P_n.
\end{equation}
With respect to a positive weight function $w(z)$, the pairing
\begin{equation}
    \langle f,g\rangle =\int f(z)g(z)w(z)dz
\end{equation}
defines an inner product on the space of polynomials and the family $\{\pi_k\}_{k\geq0,k\neq 1}$ can be obtained by Gram-Schmidt procedure. Due to the particular choice of the $\sf A$-cycle $\gamma$, and since  $\wp$ and $\wp'$ are real on $\gamma$, we obtain that the polynomials $\pi_n$ are real on the support $\gamma$. 



\subsection{Overview of the results}

Our first set of results are Theorem \ref{thm:five-term rec} and \ref{thm:sevent} which establish a five-term and a seven-term recurrence relation respectively for EOPs with generic degree-independent weights. In fact, we observe that these relations are coupled. We make this precise by proving an elliptic polynomial analogue of the Shohat-Favard theorem in Section \ref{subsec:ShohatFavard}, wich recovers the class of EOPs using the coupled recurrence relations.  Next, with the results of Section~\ref{sec:rec-rel}, we establish a Christoffel--Darboux formula for EOPs in Theorem \ref{thm:CD}. This construction leads to a natural determinantal point process supported on a cycle on the torus. We further express the CD kernel in terms of the associated RHP \ref{RHPY} in Proposition \ref{prop:CD-RHP}. 

As a final part of the paper, in Section \ref{sec:symm-weight}, we consider the case where the weight function is symmetric. Under this assumption, the orthogonal polynomials split into families of odd and even functions. This structural symmetry reduces the five-term recurrence relation to a three-term recurrence relation for odd and even polynomials respectively, while the seven-term relation reduces to a four-term relation which describes the odd polynomials in terms of even ones and vice versa. This phenomenon is described in Propositions \ref{prop:3t} and \ref{prop:4t}. The corresponding CD formula also simplifies significantly as shown in \ref{subsec:CD-even}. In fact, it turns out to be the elliptic analogue, under the transformation $z\to \wp(z)$ of the classical CD formula for OPs on the complex plane. For completeness, we demonstrate how the zeros of the even EOPs correspond to the eigenvalues of certain Jacobi matrices, and we establish interlacing properties of these zeros in Section \ref{subsec: Jacobi-even}. Finally, we also provide a Heine type formula of even EOPs in Section \ref{subsec:Heine-even}. While for EOPs with non-symmetric weights, the interlacing of zeros and Heine type formulas are not accessible with present methods, we hope to obtain insights from the case of even EOPs.

Beyond their intrinsic analytic interest, these results suggest promising applications in the spectral theory of integrable systems and elliptic generalizations of random matrix ensembles.

\subsection{Outlook}\label{subsec:outloook}
Several interesting questions remain open following our work. We detail few of them.
\begin{itemize}
    \item As mentioned above, one can aim to understand the interlacing properties of the zeros of EOPs for non-symmetric weights. From the five term recurrence relation, we obtain that the corresponding Jacobi matrix is symmetric, pentadiagonal, and the Christoffel-Darboux formula is more involved than usual. It would be interesting to find an analogous result in the elliptic situation.
   
    \item  While our main aim here was to derive the CD formula, the exploration of the associated determinantal point process (DPP) is a natural next step. We describe the existence of such a process, and an in-depth analysis of such a DPP and its relation to the results of \cite{bertola2022critical} will be carried out in a future study.

    \item The expression of the CD kernel in terms of the solution of the RHP in \eqref{prop:CD-RHP} and the canonical determinantal point process constructed in Section \ref{sec:CD} paves the way for the double scaling limit of the  kernel and could lead to novel universality results. One can achieve this in a tedious yet straightforward manner by adapting the nonlinear steepest descent analysis in \cite{bertola2022nonlinear} in an appropriate manner. For the EOPs considered here, the weight is degree-independent which would lead to absence of local parametrices. A more challenging and interesting questions comes when one is able to generalise the present story to include $n$-dependent weights, for which the associated RHP is not known in the literature. We plan to address this questions in future.
    \item Finally, it is important to note that there exist at least two other notions of CD kernels on Riemann surfaces, in context of matrix valued orthogonal polynomials (MVOPs) \cite{charlier2021matrix}, and bi-orthogonal matrix valued polynomials \cite{bertola2023abelianization, charlier2021matrix}. These papers hinge on the relation between matrix and scalar valued orthogonality. It would be interesting to relate the notion of orthogonal polynomials on Riemann surfaces \cite{bertola2022nonlinear, DLR24} to MVOPs. If such a connection exists, it could yield new insights into lattice models \cite{groot2021matrix} and may help relate the Christoffel-Darboux kernel studied in this paper to those appearing in the matrix-valued setting.
    
\end{itemize}
\section{Recurrence relations for EOPs for generic weight}\label{sec:rec-rel}
In this section, we will first show that the EOPs $\pi_n(z)$ satisfy five-term and seven-term recurrence relations. We elaborate on the relation between the two expressions and explain the relation between them. Finally, we present an elliptic version of Shohat-Favard theorem, which constructs a family of EOPs given a recurrence relation. Consequently, we show that the five-term recurrence reproduces polynomials $\pi_n(z)$ for $n\geq 3$, while the seven-term recurrence relation is needed to construct EOPs for $n\geq 0$. This reflects our definition of EOPs where $\pi_1(z)=0$.

\subsection{Coupled recurrence relations}
We now derive the five and seven-term recurrence relaitons individually.
\subsubsection{Five-term recurrence relation}
\begin{theorem}\label{thm:five-term rec}
Consider the family of EOPs in Definition \ref{def:orthonormalEOPs}. They satisfy a five-term recurrence relation for $n\geq 2$
\begin{equation}\label{eq:5term}
\begin{aligned}
    \wp(z)\pi_n(z) =\ & a_{n+1}\pi_{n+2}(z) + b_{n+1}\pi_{n+1}(z) + c_n\pi_n(z) \\
    &+ b_n\pi_{n-1}(z) + a_{n-1}\pi_{n-2}(z),
\end{aligned}
\end{equation}
where the coefficient $c_{k,n}$ is defined by the orthogonality of the polynomials in \eqref{eq:coeffs},
\begin{align}\label{def:coeffsabc}
a_{n+1} = c_{n+2,n} = c_{n,n+2}, && b_{n+1} = c_{n+1,n} = c_{n,n+1}, &&  c_n = c_{n,n},
\end{align}
and \( a_{n} > 0 \).
\end{theorem}

\begin{proof}
Let \( \pi_n(z) \) be an orthonormal polynomial of degree \( n \). Since \( \wp(z) \) is an even elliptic function with a double pole at \( z = 0 \), the product \( \wp(z)\pi_n(z) \) is a meromorphic function on the torus with a pole of order \(n+ 2 \) at \( z = 0 \), and is therefore an elliptic polynomial of degree \( n + 2 \).

As such, \( \wp(z)\pi_n(z) \) can be expanded in terms of the orthonormal basis \( \{ \pi_k(z) \} \), up to degree \( n+2 \):
\begin{equation}\label{eq:expansion}
    \wp(z)\pi_n(z) = \sum_{k=0}^{n+2} c_{k,n} \pi_k(z),
\end{equation}
where the coefficients \( c_{k,n} \) are defined by orthogonality
\begin{equation}\label{eq:coeffs}
    c_{k,n}  := \int_{\gamma} \wp(z)\pi_n(z)\pi_k(z) \, w(z)\, dz.
\end{equation}
Now observe that for any \( k < n - 2 \), the function \( \wp(z)\pi_k(z) \) has a degree strictly less than \( n \), and by orthogonality, the right-hand side of \eqref{eq:coeffs} vanishes for these values. This implies that \( c_{k,n} = 0 \) for \( k < n - 2 \). Consequently, the expansion \eqref{eq:expansion} reduces to
\begin{equation}\label{eq:5term-expanded}
\begin{aligned}
    \wp(z)\pi_n(z) =\ & c_{n+2,n}\pi_{n+2}(z) + c_{n+1,n}\pi_{n+1}(z) + c_{n,n}\pi_n(z) \\
    &+ c_{n-1,n}\pi_{n-1}(z) + c_{n-2,n}\pi_{n-2}(z).
\end{aligned}
\end{equation}

We now show the symmetry of the coefficients in the expression above. Observe that
\begin{align}
    c_{n+2,n} = \int_{\gamma} \wp(z)\pi_n(z)\pi_{n+2}(z)\, w(z)\, dz = \int_{\gamma} \wp(z)\pi_{n+2}(z)\pi_n(z)\, w(z)\, dz = c_{n,n+2},
\end{align}
by the symmetry of the integrand. Hence we may define:
\[
    a_{n+1} := c_{n+2,n} = c_{n,n+2}, \quad b_{n+1} := c_{n+1,n} = c_{n,n+1},
\]
which gives us the simplified five-term recurrence relation \eqref{eq:5term}.

Finally, we determine the sign of \( a_{n+1} \). Since the orthonormal polynomials \( \pi_n \) behave as
\[
    \pi_n(z) = \frac{z^{-n}}{\sqrt{h_n}}(1 + \mathcal{O}(z)) \quad \text{for } z \to 0,
\]
and \( \wp(z) = z^{-2} + \mathcal{O}(1) \) near \( z = 0 \), it follows that:
\[
    \wp(z)\pi_n(z) = \frac{z^{-n-2}}{\sqrt{h_n}}(1 + \mathcal{O}(z)).
\]
Matching the above asymptotic behaviour to the leading order term of the expansion \eqref{eq:5term}, which comes from \( a_{n+1}\pi_{n+2}(z) \), we must have
\[
    \frac{1}{\sqrt{h_n}} = a_{n+1} \frac{1}{\sqrt{h_{n+2}}} \quad \Rightarrow \quad a_{n+1} = \sqrt{\frac{h_{n+2}}{h_n}} > 0.
\]
Consequently, $a_n>0$. This concludes the proof.
\end{proof}

We can simplify the five-term relation further. Let us introduce the vector-valued polynomial
\[
\Pi_{2n}(z) = 
\begin{bmatrix}
    \pi_{2n+1}(z) \\
    \pi_{2n}(z)
\end{bmatrix}.
\]
Then the five-term recurrence relation~\eqref{eq:5term} may be recast as a three-term matrix recurrence of the form
\begin{equation} \label{eq:matrix-recurrence}
    \wp(z)\Pi_{2n}(z) = \mathcal{A}_{2n+2} \Pi_{2n+2}(z) + \mathcal{B}_{2n} \Pi_{2n}(z) + \mathcal{A}_{2n}^{\mathsf{T}} \Pi_{2n-2}(z),
\end{equation}
where the recurrence matrices are explicitly given by
\begin{equation} \label{eq:rec-matrices}
    \mathcal{A}_{2n+2} = 
    \begin{bmatrix}
        a_{2n+2} & b_{2n+2} \\
        0 & a_{2n+1}
    \end{bmatrix}, 
    \qquad
    \mathcal{B}_{2n} = 
    \begin{bmatrix}
        c_{2n+1} & b_{2n+1} \\
        b_{2n+1} & c_{2n}
    \end{bmatrix}.
\end{equation}

\subsubsection{Seven-term recurrence relation}

As opposed to the above description, we now consider the expression $\wp'(z)\pi_{n}(z)$ which is a polynomial of degree in $n+3$. We expand $\wp'(z)\pi_{n}(z)$ in the basis of orthonormal polynomials
\begin{equation}\label{expansion2}
     \wp'(z)\pi_{n}(z)=\sum_{i=0}^{n+3}c_{i,n}\pi_{i}(z) 
\end{equation}
For $i<n-3$ we have $\wp'(z)\pi_{i}(z)$ is a polynomial of degree $i+3<n$. By a similar argument as above, we have the following theorem.
\begin{theorem}\label{thm:sevent}
Let \( \{\pi_n(z)\}_{n \geq 0,\, n \neq 1} \) be a family of orthonormal elliptic polynomials with respect to a weight function \( w(z) \) on a contour \( \gamma \). Then multiplication by the derivative of the Weierstrass function, \( \wp'(z) \), yields the seven-term recurrence relation:
\begin{equation}\label{eq:7term}
\begin{aligned}
\wp'(z)\pi_n(z) =\ &\, p_{n+3}\pi_{n+3}(z) + q_{n+2}\pi_{n+2}(z) + r_{n+1}\pi_{n+1}(z) \\
&+ s_n\pi_n(z) + r_n\pi_{n-1}(z) + q_n\pi_{n-2}(z) + p_n\pi_{n-3}(z),
\end{aligned}
\end{equation}
where \( p_{n+3} = -2\sqrt{h_{n+3}/h_n} < 0 \), and
\[
\int_{\gamma} \wp'(z)\, \pi_n(z)\pi_k(z)\, w(z)\, dz =
\begin{cases}
q_{n+2}, & k = n+2, \\
r_{n+1}, & k = n+1, \\
s_n,     & k = n,\\
r_n      & k=n-1,\\
q_n &k=n-2.\\
\end{cases}
\]
\end{theorem}
\begin{remark}
    It is important to note that for EOPs, we need {\it both} the five and seven-term recurrence relations simultaneously. The interplay between the two relations can be seen through the simple case of $n=0$ for \eqref{7t} and $n=1$ for \eqref{5t}, which define a simultaneous system of linear equations that determine $\pi_{2}(z)$ and $\pi_3(z)$. The derivation of higher degree polynomials then follows.
\end{remark}
\begin{remark}
    The coefficients of seven-term and five-term are related in a non-trivial way by the elliptic curve. 
    A trivial condition is that 
    \begin{equation}
       4 a_{n+5}a_{n+3}a_{n+1}=p_{n+6}p_{n+3}.
    \end{equation}
    We refer the reader to Appendix \ref{App:rel_rec_coeff} for a detailed list of relations.
\end{remark}

\subsection{Elliptic analogue of the Shohat-Favard Theorem}\label{subsec:ShohatFavard}
At first glance, one might question the necessity of the seven-term recurrence when the five-term recurrence already exists. The seven-term recurrence, however, plays an equally crucial role. The reason for this lies in the construction of the polynomials from the recurrence coefficients. While the five-term recurrence allows for the construction of polynomials up to a certain degree, the polynomial \( \pi_3 \) cannot be generated using just the five-term recurrence because the polynomial \( \pi_1 \) is undefined. By adding the seven-term recurrence, we can generate polynomials of any degree, even from the sequence \( \pi_2 \) and \( \pi_3 \). This ensures the consistency of the recurrence coefficients and ensures that the polynomials are well-defined. In fact, the orthogonality measure becomes uniquely determined when the recurrence coefficients satisfy this condition. This fact is rigorously addressed in the elliptic version of the Shohat-Favard theorem, which we state and prove below.
\begin{definition}
    Let $\mathcal{P}$ be the vector space of all elliptic polynomials. We define a moment functional to be a linear functional $\mathcal{L}:\mathcal{P}\rightarrow\mathbb R$. We say $\mathcal{L}$ is positive definite if for all $f\in \mathcal{P}$ we have $\mathcal{L}(f^2)>0$.
\end{definition}
{In what follows we show that if we inductively build our family of polynomials according to the five and seven term recurrence relation consistently, then there exists a moment functional, for which the polynomials are orthogonal.
With the proof that the functions satisfying the recurrence relation are orthogonal w.r.t to a weight function, one obtains the representation of the coefficients such as \eqref{eq:coeffs} and \eqref{def:coeffsabc}.}
\begin{theorem}
Let \( \{a_n>0\}, \{b_n\}, \{c_n\}, \{p_n>0\}, \{q_n\}, \{r_n\} \{s_n\}\) be sequences, and let \( \lambda_1>0 \) be a constant. According to our definition of EOPs (see Def \ref{def:orthonormalEOPs}), we further assume  $\pi_1(z)=0$, and $\pi_{k}=0$ for $k<0$. 
Consider the five and seven-term recurrence relations for the polynomials $\pi_k(z)$ respectively

\begin{equation}\label{5t}
   - a_{n-1} \pi_n(z) = b_{n-1} \pi_{n-1}(z) + (c_{n-2} - \wp(z)) \pi_{n-2}(z) + b_{n-2} \pi_{n-}(z) + a_{n-3} \pi_{n-4}(z),
\end{equation}
and
\begin{equation}\label{7t}
\begin{aligned}
-p_n \pi_n(z) = q_{n-1} \pi_{n-1}(z) + r_{n-2} \pi_{n-2}(z) + ( s_{n-3}-&\wp'(z)) \pi_{n-3}(z) + r_{n-3} \pi_{n-4}(z)\\& + q_{n-3} \pi_{n-5}(z) + p_{n-3} \pi_{n-6}(z).
\end{aligned}\end{equation}
 If the system is consistent, then there exists a unique positive definite moment functional \( \mathcal{L}: \mathcal{P} \to \mathbb{R} \) such that:
\[
\mathcal{L}(1) = \lambda_1, \quad \mathcal{L}(\pi_n(z) \pi_m(z)) = \lambda_n \delta_{m,n},
\quad \lambda_n>0.\]
\end{theorem}
\begin{proof}
We define the moment functional \( \mathcal{L} \) on the space of polynomials by specifying its action on the orthogonal polynomial sequence \( \{\pi_n(z)\}_{n \geq 0} \) as follows:
\[
\mathcal{L}(1) = \lambda_1, \qquad \mathcal{L}(\pi_n(z)) = 0 \quad \text{for all } n \geq 1.
\]
This definition extends uniquely to all polynomials via linearity.

Let \( n = 2k \) be even, and fix an integer \( i \leq k - 1 \). From the five-term recurrence relation satisfied by the orthogonal polynomials (denoted as equation~\eqref{5t}), we can write:
\[
\wp(x)^i \pi_{2k}(z) = \sum_{j = 2k - 2i}^{2k + 2i} d_j \pi_j(z), \quad \text{for some } d_j \in \mathbb{R}.
\]
Applying \( \mathcal{L} \) to both sides and using the defining property that \( \mathcal{L}(\pi_j(z)) = 0 \) for \( j \geq 1 \), we find:
\[
\mathcal{L}(\wp(x)^i \pi_{2k}(z)) = \sum_{j = 2k - 2i}^{2k + 2i} d_j \mathcal{L}(\pi_j(z)) = 0.
\]

Now consider the expression \( \mathcal{L}(\wp'(z) \wp^i(z) \pi_{2k}(z)) \) with $i\leq k-2$. Using the same recurrence, we write:
\begin{equation}\label{eq:sum_wpprime}
\mathcal{L}(\wp'(z) \wp^i(z) \pi_{2k}(z)) = \sum_{j = 2k - 2i}^{2k + 2i} \mathcal{L}(\wp'(z) d_j \pi_j(z)).
\end{equation}
Each term of the form \( \wp'(z) \pi_j(z) \) can be expressed using the seven-term recurrence relation ~\eqref{7t}. For example, consider the highest-degree term for \( i = k - 2 \), namely:
\[
\wp'(z) \pi_4(z) = \sum_{j = 2}^{7} e_j \pi_j(z), \quad \text{for some } e_j \in \mathbb{R}.
\]
Then,
\[
\mathcal{L}(d_4 \wp'(z) \pi_4(z)) = \sum_{j = 2}^{7} e_j \mathcal{L}(\pi_j(z)) = 0.
\]
As all such terms in~\eqref{eq:sum_wpprime} vanish by the same reasoning, we conclude:
\[
\mathcal{L}(\wp'(z) \wp^i(z) \pi_{2k}(z)) = 0 \quad \text{for all } i \leq k - 2.
\]

A completely analogous argument applies in the case where \( n = 2k + 1 \) is odd. Thus, for any \( m < n \), the inner product induced by \( \mathcal{L} \) satisfies:
\[
\mathcal{L}(\pi_m(z) \pi_n(z)) = 0,
\]
which shows that \( \{\pi_n(z)\} \) is an orthogonal system under \( \mathcal{L} \).

To study the moments, observe that applying \( \mathcal{L} \) to the recurrence for \( \pi_{2n}(z) \), we get:
\[
\mathcal{L}(\wp(z)^n \pi_{2n}(z)) = a_{2n-1} \mathcal{L}(\wp(z)^{n-1} \pi_{2n-2}(z)).
\]
This recurrence allows us to compute moments recursively from lower degrees.

Similarly,
\[
\mathcal{L}(\wp'(z) \wp^n(z) \pi_{2n+3}(z)) = a_{2n+2} \mathcal{L}(\wp'(z) \wp^{n-1}(z) \pi_{2n+1}(z)),
\]

Since the action of \( \mathcal{L} \) vanishes on all polynomials orthogonal to the constant function and satisfies a recursive moment structure on the diagonal, the orthogonality is complete and well-posed. Furthermore, assuming the recurrence coefficients \( a_n, p_n > 0 \), the functional \( \mathcal{L} \) is also positive definite.
\end{proof}

\section{Christoffel-Darboux formula}\label{sec:CD}
In this section, we first use the five-term recurrence relation in Theorem \ref{thm:five-term rec} to obtain a Christoffel-Darboux (CD) formula for EOPs. We then express it in terms of the solution of the RHP \ref{RHPY}. 

\begin{theorem}[Christoffel-Darboux Formula for Elliptic Orthogonal Polynomials]\label{thm:CD}
Consider the EOPs $\pi_n(z)$ described in Definition \ref{def:orthonormalEOPs}. Then the Christoffel-Darboux kernel
\[
\widehat{K_n}(x,y) := \sum_{j=0}^{n-2} \pi_j(x)\pi_j(y)
\]
is given by
\begin{equation} \label{eq:cd}
\begin{aligned}
\widehat{K_n}(x,y) = \frac{1}{\wp(x) - \wp(y)} \bigg[&a_{n-1} \left( \pi_n(x)\pi_{n-2}(y) - \pi_n(y)\pi_{n-2}(x) \right) \\
&+ a_{n-2} \left( \pi_{n-1}(x)\pi_{n-3}(y) - \pi_{n-1}(y)\pi_{n-3}(x) \right) \\
&+ b_{n-1} \left( \pi_{n-1}(x)\pi_{n-2}(y) - \pi_{n-1}(y)\pi_{n-2}(x) \right) \bigg],
\end{aligned}
\end{equation}
where the coefficients $a_{k}, b_{k}$ are obtained through orthogonality condition as defined in \eqref{def:coeffsabc}.
\end{theorem}

\begin{proof}
We begin by using the five-term recurrence relation for the action of \( \wp(z) \) on $\pi_n(z)$ in \eqref{5t}:
\begin{equation} \label{eq:five-term}
\wp(z)\pi_j(z) = a_{j+1} \pi_{j+2}(z) + b_{j+1} \pi_{j+1}(z) + c_j \pi_j(z) + b_j \pi_{j-1}(z) + a_{j-1} \pi_{j-2}(z),
\end{equation}
where the coefficients \( a_j, b_j, c_j \in \mathbb{R} \), and \( a_j > 0 \) for all \( j \). Multiplying both sides of \eqref{eq:five-term} by \( \pi_j(y) \) and summing over \( j \) from 0 to \( n-2 \), we obtain
\begin{align*}
\sum_{j=0}^{n-2} \wp(x)\pi_j(x)\pi_j(y) 
&= \sum_{j=0}^{n-2} \left( a_{j+1} \pi_{j+2}(x) + b_{j+1} \pi_{j+1}(x) + c_j \pi_j(x) + b_j \pi_{j-1}(x) + a_{j-1} \pi_{j-2}(x) \right) \pi_j(y).
\end{align*}
The right-hand side of the above expression can be rearranged by collecting terms that are symmetric in \( x \) and \( y \):
\begin{align}
\sum_{j=0}^{n-2} \wp(x)\pi_j(x)\pi_j(y)
&= \sum_{j=0}^{n-2} c_j \pi_j(x)\pi_j(y) 
+ \sum_{j=1}^{n-2} b_{j-1} \left( \pi_{j-1}(x)\pi_j(y) + \pi_j(x)\pi_{j-1}(y) \right) \\
&\quad + \sum_{j=2}^{n-2} a_{j-1} \left( \pi_{j-2}(x)\pi_j(y) + \pi_j(x)\pi_{j-2}(y) \right) \\
&\quad + a_{n-1} \pi_n(x)\pi_{n-2}(y) + a_{n-2} \pi_{n-1}(x)\pi_{n-3}(y) + b_{n-1} \pi_{n-1}(x)\pi_{n-2}(y).\label{eq:expCDin}
\end{align}
Similarly, one obtains a similar expression for  \( x \leftrightarrow y \)
\begin{align}
\sum_{j=0}^{n-2} \wp(y)\pi_j(y)\pi_j(x)
&= \sum_{j=0}^{n-2} c_j \pi_j(y)\pi_j(x) 
+ \sum_{j=1}^{n-2} b_{j-1} \left( \pi_{j-1}(y)\pi_j(x) + \pi_j(y)\pi_{j-1}(x) \right) \\
&\quad + \sum_{j=2}^{n-2} a_{j-1} \left( \pi_{j-2}(y)\pi_j(x) + \pi_j(y)\pi_{j-2}(x) \right) \\
&\quad + a_{n-1} \pi_n(y)\pi_{n-2}(x) + a_{n-2} \pi_{n-1}(y)\pi_{n-3}(x) + b_{n-1} \pi_{n-1}(y)\pi_{n-2}(x).\label{eq:expCDinsymm}
\end{align}
Subtracting the expressions \eqref{eq:expCDin} and \eqref{eq:expCDinsymm} we get
\begin{align*}
(\wp(x)-\wp(y))\sum_{j=0}^{n-2} \pi_j(y)\pi_j(x) &=a_{n-1} \left( \pi_n(x)\pi_{n-2}(y) - \pi_n(y)\pi_{n-2}(x) \right)\\
&+ a_{n-2} \left( \pi_{n-1}(x)\pi_{n-3}(y) - \pi_{n-1}(y)\pi_{n-3}(x) \right)\\
&+ b_{n-1} \left( \pi_{n-1}(x)\pi_{n-2}(y) - \pi_{n-1}(y)\pi_{n-2}(x) \right).
\end{align*}
Dividing both sides by \( \wp(x) - \wp(y) \) gives the result in \eqref{eq:cd}.
\end{proof}

\begin{corollary}[Confluent Christoffel-Darboux Formula]
Taking the limit \( y \to x \) in \eqref{eq:cd}, we obtain
\begin{equation}
\begin{aligned}
\sum_{j=0}^{n-2} \pi_j(x)^2 = \frac{1}{\wp'(x)} \bigg[ &a_{n-1} \left( \pi_n(x)\pi_{n-2}'(x) - \pi_n'(x)\pi_{n-2}(x) \right) \\
&+ a_{n-2} \left( \pi_{n-1}(x)\pi_{n-3}'(x) - \pi_{n-1}'(x)\pi_{n-3}(x) \right) \\
&+ b_{n-1} \left( \pi_{n-1}'(x)\pi_{n-2}(x) - \pi_{n-1}(x)\pi_{n-2}'(x) \right) \bigg].
\end{aligned}
\end{equation}
\end{corollary}

\begin{corollary}[Degenerate Points]
If \( x \in \{\tau/2, (1+\tau)/2\} \), where \( \wp'(x) = 0 \), then the confluent formula takes the form
\begin{equation}
\begin{aligned}
\sum_{j=0}^{n-2} \pi_j(x)^2 = \frac{1}{\wp''(x)} \bigg[ &a_{n-1} \left( \pi_n(x)\pi_{n-2}''(x) - \pi_n''(x)\pi_{n-2}(x) \right) \\
&+ a_{n-2} \left( \pi_{n-1}(x)\pi_{n-3}''(x) - \pi_{n-1}''(x)\pi_{n-3}(x) \right) \\
&+ b_{n-1} \left( \pi_{n-1}''(x)\pi_{n-2}(x) - \pi_{n-1}(x)\pi_{n-2}''(x) \right) \bigg].
\end{aligned}
\end{equation}
\end{corollary}
We now express the CD formula in terms of the solution of the RHP \ref{RHPY} in the following proposition. For convenience, we present this result in terms of the EOPs $P_n$.
\begin{proposition}\label{prop:CD-RHP}
    Recall the Riemann-Hilbert problem \ref{RHPY}, then the Christoffel-Darboux kernel can expressed in the following form,
\begin{equation}\label{RHPCD}
\begin{aligned}
&\widehat{K_n}(x, y) \\
&= \, \frac{a_{n-1} h_{n-2}}{2\pi i (\wp(x) - \wp(y))} \begin{pmatrix} 0 & 1 \end{pmatrix} \left( \det Y_{n-1}(y) Y_{n-1}^{-1}(y) Y_n(x) - \det Y_{n-1}(x) Y_{n-1}^{-1}(x) Y_n(y) \right) \begin{pmatrix} 1 \\ 0 \end{pmatrix} \\
& + \frac{a_{n-2} h_{n-3}}{2\pi i (\wp(x) - \wp(y))} \begin{pmatrix} 0 & 1 \end{pmatrix} \left( \det Y_{n-2}(y) Y_{n-2}^{-1}(y) Y_{n-1}(x) - \det Y_{n-2}(x) Y_{n-2}^{-1}(x) Y_{n-1}(y) \right) \begin{pmatrix} 1 \\ 0 \end{pmatrix} \\
& + \frac{b_{n-1} h_{n-1}}{2\pi i (\wp(x) - \wp(y))} \begin{pmatrix} 0 & 1 \end{pmatrix} \left( \det Y_{n-1}(y) Y_{n-1}^{-1}(y) Y_{n-1}(x) \right) \begin{pmatrix} 1 \\ 0 \end{pmatrix}.
\end{aligned}
\end{equation}
\end{proposition}
    \begin{proof}
It was shown in \cite[Lemma~2.3]{DLR24} that
\[
\det Y_n(z) = \wp(z) + c_n,
\]
where $\wp(z)$ is the Weierstrass elliptic function and $c_n$ is a constant depending on $n$. Since $\wp(z) + c_n$ is a doubly periodic meromorphic function with a double pole and two simple zeros in the fundamental region, it follows that $\det Y_n(z)$ has exactly two simple zeros, which we denote by $x_1(n)$ and $x_2(n)$.

Away from the points $0$, $x_1(n)$, and $x_2(n)$, the matrix $\det Y_n(z)$ is nonzero and thus $Y_n(z)$ is invertible. At $z = x_k(n)$ for $k = 1, 2$, we define the expression
\[
\det Y_{n-1}(z) \cdot Y_{n-1}^{-1}(z) := \lim_{\zeta \to z} \det Y_{n-1}(\zeta) \cdot Y_{n-1}^{-1}(\zeta),
\]
which we interpret as the analytic continuation of the matrix product from a punctured neighborhood of $z$. In particular, we define
\[
\det Y_{n-1}(z) \cdot Y_{n-1}^{-1}(z) := 
\begin{pmatrix}
\frac{2\pi i}{h_{n-1}} \mathcal{C}(P_{n-1})(z) & -\frac{2\pi i}{h_{n-1}} P_{n-1}(z) \\
-\mathcal{C}(P_n)(z) & P_n(z)
\end{pmatrix},
\quad z \in \{x_1(n), x_2(n)\}.
\]

A direct computation using the solution of the RHP \eqref{RHPsol} shows that the Christoffel–Darboux expression
\begin{align}
&\frac{a_{n-1} \left( \pi_n(x) \pi_{n-2}(y) - \pi_n(y) \pi_{n-2}(x) \right)}{\wp(x) - \wp(y)}\\
&=
\frac{a_{n-1} h_{n-2}}{2\pi i (\wp(x) - \wp(y))} 
\begin{pmatrix} 0 & 1 \end{pmatrix}
\left( 
\det Y_{n-1}(y) \cdot Y_{n-1}^{-1}(y) Y_n(x) 
- 
\det Y_{n-1}(x) \cdot Y_{n-1}^{-1}(x) Y_n(y) 
\right)
\begin{pmatrix} 1 \\ 0 \end{pmatrix},\label{t1}
\end{align}
and
 \begin{align}
&\frac{a_{n-2} \left( \pi_{n-1}(x) \pi_{n-3}(y) - \pi_{n-1}(y) \pi_{n-3}(x) \right)}{\wp(x) - \wp(y)}\\
&=\frac{a_{n-2} h_{n-3}}{2\pi i (\wp(x) - \wp(y))} \begin{pmatrix} 0 & 1 \end{pmatrix} \left( \det Y_{n-2}(y) Y_{n-2}^{-1}(y) Y_{n-1}(x) - \det Y_{n-2}(x) Y_{n-2}^{-1}(x) Y_{n-1}(y) \right) \begin{pmatrix} 1 \\ 0 \end{pmatrix}.
 \end{align}
Similarly, we obtain
\begin{align}
&\frac{b_{n-1} \left( \pi_{n-1}(x)\pi_{n-2}(y) - \pi_{n-1}(y)\pi_{n-2}(x) \right)}{\wp(x)-\wp(y))}\\
&=\frac{b_{n-1} h_{n-1}}{2\pi i (\wp(x) - \wp(y))} \begin{pmatrix} 0 & 1 \end{pmatrix} \left( \det Y_{n-1}(y) Y_{n-1}^{-1}(y) Y_{n-1}(x) \right) \begin{pmatrix} 1 \\ 0 \end{pmatrix}.
\end{align}
Substituting the above expressions in \eqref{eq:cd} gives \eqref{RHPCD}.
\end{proof}
\begin{remark}[Spectral Interpretation of Zeros]
For our EOPs, we have $\pi_1(z)=0$. Now set $a_2 = b_1 = b_2 = 0$, $c_1 = 1$. From the five-term recurrence coefficients, we can then construct the infinite, symmetric, penta-diagonal matrix 
\[
J = \begin{pmatrix}
c_0 & b_1 & a_1 & 0 & 0 & 0 & \cdots \\
b_1 & c_1 & b_2 & a_2 & 0 & 0 & \cdots \\
a_1 & b_2 & c_2 & b_3 & a_3 & 0 & \cdots \\
0 & a_2 & b_3 & c_3 & b_4 & a_4 & \cdots \\
0 & 0 & a_3 & b_4 & c_4 & b_5 & \cdots \\
\vdots & \vdots & \vdots & \vdots & \vdots & \ddots & \ddots
\end{pmatrix}.
\]
Let $J_{n,n+1}$ denote the $n \times (n+1)$ top-left submatrix of $J$. Then, the recurrence relation implies the following linear identity for the column vector of orthonormal polynomials:
\[
J_{n,n+1}
\begin{pmatrix}
\pi_0(x) \\
\pi_1(x) \\
\vdots \\
\pi_n(x)
\end{pmatrix}
=
\wp(x)
\begin{pmatrix}
\pi_0(x) \\
\pi_1(x) \\
\vdots \\
\pi_{n-1}(x)
\end{pmatrix}
-
a_n \pi_{n+1}(x) \, e_n,
\]
where $e_n$ is the $n$-th standard basis vector in $\mathbb{C}^n$. In particular, if $x_0$ is a zero of $\pi_{n+1}(x)$, then the correction term vanishes, and we obtain
\[
J_{n,n+1} v_n = \wp(x_0) \cdot P v_n,
\]
where $v_n = (\pi_0(x_0), \dots, \pi_n(x_0))^T$ and $P$ is the projection onto the first $n$ coordinates. Thus, the zeros of $\pi_{n+1}$ correspond to the compressed eigenvalue problem for $J_{n,n+1}$ under the projection $P$.
\end{remark} 

With the CD kernel constructed above, we define the new correlation kernel 
\begin{equation}\label{DPPKer:def}
    K_n(x,y):=\sqrt{w(x)}\sqrt{w(y)}\widehat{K}_{n+1}(x,y)
\end{equation}

\begin{proposition} The kernel $K(x,y)$ has the following properties.
    \begin{enumerate}
        \item Positivity: For $x_{i}\in \gamma$ we have
    \begin{equation}
        \det [K(x_i,x_j)]_{i,j=1}^{n}= \left(\det \left[\sqrt{w(x_j)}\pi_{k-1}(x_j)\right]_{j,k=1}^{n}\right)^2\geq 0 
    \end{equation}
    \item Trace-class: $$\mathrm{Tr} K = \int_{\gamma}K_{n}(x,x)=\sum_{i=0}^{n-1}\int_{\gamma}w(x)\pi_{k}(x)\pi_{k}(x)dx=n$$
    \item Reproducing property:$$\int_{\gamma}K_{n}(x,s)K_n(s,y)ds=K_n(x,y).$$
 \end{enumerate}
\end{proposition}
The proof follows through the classical arguments in \cite{forrester2010log}.


The advantage of introducing the kernel \eqref{DPPKer:def} is that it defines a point process. Indeed, it is known that the determinant of a kernel with the properties above defines a symmetric probability density function (p.d.f)
\begin{align}\label{eq:pdf}
    P(x_1,\dots, x_n) = \frac{1}{n!} \det\left[K(x_i, x_j) \right]_{i,j=1}^{n}.
\end{align}
In our case, this is associated to a determinantal point process on the interval $\gamma^n$, with $K(x_i,x_j)$ assuming the role of the correlation kernel. 

The above fact merits a separate study. As mentioned in Section \ref{subsec:outloook}, we expect interesting phenomena with $n$-dependent weights, and we reserve a detailed study of point processes associated to EOPs with general weights for a future work.

\section{Case of the symmetric weight}\label{sec:symm-weight}
With the general setup of the previous sections, we now turn to a special case of the weight function, namely when $w(x)$ is symmetric 
\begin{equation}\label{eq:even_weight}
    {w}\left(\tfrac{1}{2}(1+\tau) + z\right) = {w}\left(\tfrac{1}{2}(1+\tau) - z\right), \quad z \in [0,1].
\end{equation}


In this case, the system of EOPs decompose naturally into families of even and odd degree polynomials \cite{DLR24}.
More precisely, the monic orthogonal polynomials \( \pi_n(x, \tau) \) admit the following explicit forms:
\begin{itemize}
    \item For even degrees \( n = 2k \), one has
    \begin{equation}
       \pi_{2k}(z,\tau) = \sum_{i=0}^k a_{i,2k}(\tau)\, \wp(z)^i, \qquad \text{with} \quad a_{k,2k}(\tau) = 1.
    \end{equation}

    \item For odd degrees \( n = 2k+3 \), one has
    \begin{equation}
      \pi_{2k+3}(z,\tau) = -\frac{1}{2} \wp'(z) \sum_{i=0}^k a_{i,2k+3}(\tau)\, \wp(z)^i, \qquad \text{with} \quad a_{k,2k+3}(\tau) = 1.
    \end{equation}
\end{itemize}
This splitting reflects the underlying parity symmetry induced by \eqref{eq:even_weight}. We note that each of the even and odd families recover the classical OPs on the complex plane with the transformation $\wp(z) \to z$. However, studying each of these odd and even families offers insights into the dynamics of the EOPs on non-trivial topologies. In what follows, we present three results for the even polynomials namely, the interlacing property of the zeros, Heine like formula, and the CD formula. Similar expressions hold for the odd polynomials and can be obtained by repeating the arguments presented here.

\subsection{Jacobi Matrix and interlacing of the zeros}\label{subsec: Jacobi-even}
Let us begin with the three-term recurrence relation satisfied by the even-degree EOPs.
\begin{proposition}\label{prop:3t}
The orthonormal polynomials solve the following three-term recurrence relation
    \begin{equation}\label{3term2}
     \wp(z)\pi_{n}(z)=a_{n+1}\pi_{n+2}(z)+c_{n}\pi_{n}(z)+a_{n-1}\pi_{n-2}(z),
     \end{equation}
     Where the sequence $\{a_n>0\}$ and $\{c_n\}$ are defined in as of Theorem \ref{thm:five-term rec}.
\end{proposition}
\begin{remark}
The above statement can be proved starting from the five-term recurrence relation by restricting to even polynomials or through the linear system they solve, as shown in \cite{DLR24}.
\end{remark}

\begin{proposition}\label{prop:4t}
One might also obtain a four term recurrence relation
    \begin{equation}
\wp'(z)\pi_{n}(z)=p_{n+3}\pi_{n+3}(z)+r_{n+1}\pi_{n+1}(z)+r_{n}\pi_{n-1}(z)+p_{n}\pi_{n-3}(z).
\end{equation}
Where $\{p_n\}$ and $ \{r_n\}$ are defined as in Theorem \ref{eq:7term}.
\end{proposition}

\noindent
{\it Sketch of proof:} We similarly consider 

    \begin{equation}
        \wp'(z)\pi_{2n+3}(z)=\sum_{i=0}^{6}c_{2n+i,2n+3}\pi_{2n+i}(z),
    \end{equation}
where $$c_{2n+i,2n+3}=\int_{\gamma}\pi_{2n+i}(z)\pi_{2n+3}(z)\wp'(z)w(z)dz.$$ From symmetry, it follows that $c_{2n+i,2n+3}=0$ for $i=1,3,5$ and  $c_{2n+i,2n+3}=c_{2n+3,2n+i}$. We then recover the statement of the proposition.


With the recurrence relation above, the (semi-infinite) Jacobi matrix \( J \) associated with the orthogonal polynomial system \( \{\pi_{2k}(z)\}_{k=0}^\infty \) is given by
\begin{equation}
    J = \begin{pmatrix}
        \beta_0 & \alpha_1 & 0 & 0 & \dots \\
        \alpha_1 & \beta_1 & \alpha_2 & 0 & \dots \\
        0 & \alpha_2 & \beta_2 & \alpha_3 & \ddots \\
        0 & 0 & \alpha_3 & \beta_3 & \ddots \\
        \vdots & \vdots & \ddots & \ddots & \ddots
    \end{pmatrix},
\end{equation}
where \( \alpha_k > 0 \) for all \( k \geq 1 \), and \( \beta_k \in \mathbb{R} \) are the recurrence coefficients corresponding to the orthonormal polynomials. Denote by \( J_n \) the upper-left \( n \times n \) principal submatrix of \( J \). The recurrence relation can then be expressed in matrix form as
\begin{equation}
    J_n \begin{pmatrix}
        \pi_0(z) \\
        \pi_2(z) \\
        \vdots \\
        \pi_{2n-4}(z) \\
        \pi_{2n-2}(z)
    \end{pmatrix}
    =
    \wp(z) \begin{pmatrix}
        \pi_0(z) \\
        \pi_2(z) \\
        \vdots \\
        \pi_{2n-4}(z) \\
        \pi_{2n-2}(z)
    \end{pmatrix}
    -
    \begin{pmatrix}
        0 \\
        0 \\
        \vdots \\
        0 \\
        \alpha_n \pi_{2n}(z)
    \end{pmatrix}.
\end{equation}

Let \( z_0 \) be a zero of \( \pi_{2n}(z) \), i.e., \( \pi_{2n}(z_0) = 0 \). Then the residual term vanishes at \( z = z_0 \), and we obtain
\begin{equation}
    J_n \mathbf{v}(z_0) = \wp(z_0) \mathbf{v}(z_0),
\end{equation}
where the vector \( \mathbf{v}(z_0) \in \mathbb{R}^n \) is defined by
\[
\mathbf{v}(z_0) = \begin{pmatrix}
    \pi_0(z_0) \\
    \pi_2(z_0) \\
    \vdots \\
    \pi_{2n-4}(z_0) \\
    \pi_{2n-2}(z_0)
\end{pmatrix}.
\]
Hence, \( \wp(z_0) \) is an eigenvalue of the finite Jacobi matrix \( J_n \), and \( \mathbf{v}(z_0) \) is the associated eigenvector.

\begin{lemma}
Let the zeros of \( \pi_{2n}(z) \) be denoted by
\[
\{z_{1,2n}, z_{2,2n}, \dots, z_{2n-1,2n}, z_{2n,2n}\},
\]
ordered (because of symmetry) such that
\[
\wp(z_{2j,2n}) = \wp(z_{2j+1,2n}) \quad \text{for each } j = 1, \dots, n.
\]
Then, the zeros of the polynomials interlace in the following sense:
\begin{equation}
    \wp(z_{2j,2n}) < \wp(z_{2j,2n-2}) < \wp(z_{2j+2,2n}) \qquad \text{for } j = 1, \dots, n-1.
\end{equation}
Since $\wp$ is monotone in $\gamma/2$ this means they are also interlacing on $\gamma$.
\end{lemma}

\begin{proof}
The result follows from the spectral theorem applied to symmetric tridiagonal matrices, and the interlacing property of eigenvalues of principal submatrices. Since the eigenvalues of \( J_n \) correspond to the values \( \wp(z_{2j,2n}) \), and similarly those of \( J_{n-1} \) to \( \wp(z_{2j,2n-2}) \), the interlacing follows from the classical Cauchy interlacing theorem or, equivalently, from the Rayleigh quotient characterization of eigenvalues.
\end{proof}

For completeness, we now provide the Shohat-Favard theorem for even EOPs. 
\begin{proposition}
    Let \( \{\pi_{2n}(z)\}_{n \geq 0} \) be a sequence of polynomials satisfying the three-term recurrence relation
\begin{equation}\label{eq:3term_recurrence}
    \wp(z)\, \pi_{2n}(z) = a_{n+1} \pi_{2n+2}(z) + b_n \pi_{2n}(z) + a_n \pi_{2n-2}(z), \qquad n \geq 0,
\end{equation}
with initial conditions \( \pi_0(z) = w_0 \), \( \pi_{-2}(z) = 0 \), and recurrence coefficients \( a_n > 0 \), \( b_n \in \mathbb{R} \) for all \( n \geq 0 \). Then there exists a symmetric probability measure \( \mu \), supported on a compact set \( \gamma \), such that the family \( \{ \pi_{2n} \} \) is orthonormal with respect to \( \mu \):
\begin{equation}
    \int_\gamma \pi_{2n}(z)\, \pi_{2m}(z) \, d\mu(z) = \delta_{n,m}, \qquad n, m \geq 0.
\end{equation}
\end{proposition}

\begin{proof}
Let \( J_n \in \mathbb{R}^{n \times n} \) denote the \( n \times n \) truncation of the Jacobi matrix associated with the recurrence \eqref{eq:3term_recurrence}, given by
\begin{equation}
J_n = \begin{pmatrix}
    b_0 & a_1 & 0 & \cdots & 0 \\
    a_1 & b_1 & a_2 & \ddots & \vdots \\
    0 & a_2 & b_2 & \ddots & 0 \\
    \vdots & \ddots & \ddots & \ddots & a_{n-1} \\
    0 & \cdots & 0 & a_{n-1} & b_{n-1}
\end{pmatrix}.
\end{equation}
This matrix is real symmetric and tridiagonal. Hence, by the spectral theorem, it is diagonalizable with a complete orthonormal set of eigenvectors, and all its eigenvalues are real and simple.

Let \( \{z_{k,n}\}_{k=1}^n \subset \gamma/2 \) denote the zeros of \( \pi_{2n}(z) \). It follows from the recurrence relation that \( \wp(z_{k,n}) \) is an eigenvalue of \( J_n \), with associated eigenvector
\[
v_k := \left( \pi_0(z_{k,n}), \pi_2(z_{k,n}), \dots, \pi_{2n-2}(z_{k,n}) \right)^T \in \mathbb{R}^n.
\]
Let \( V \in \mathbb{R}^{n \times n} \) be the matrix whose \( k \)-th column is the normalized vector \( v_k / \|v_k\| \), and let \( D = \mathrm{diag}(\wp(z_{1,n}), \dots, \wp(z_{n,n})) \). Then \( J_n = V D V^T \), and since the eigenvectors are orthonormal, we have
\[
V^T V = I_n, \quad \text{hence} \quad VV^T = I_n.
\]

Set \( \lambda_{k,n} := 1 / \|v_k\|^2 > 0 \). Then we obtain the identity
\begin{equation}
    \sum_{k=1}^n \lambda_{k,n} \, \pi_{2i}(z_{k,n}) \, \pi_{2j}(z_{k,n}) = \delta_{i,j}, \qquad 0 \leq i,j \leq n-1.
\end{equation}
In particular, taking \( i = j = 0 \), we find
\[
\sum_{k=1}^n \lambda_{k,n} = 1.
\]
Define the discrete probability measure
\[
\mu_n := \sum_{k=1}^n \lambda_{k,n} \, \delta_{z_{k,n}},
\]
which is supported on the compact set \( \gamma/2\). Then for all \( 0 \leq i, j \leq n-1 \),
\begin{equation}
    \int_{\gamma/2} \pi_{2i}(z)\, \pi_{2j}(z) \, d\mu_n(z) = \delta_{i,j}.
\end{equation}

The sequence \( \{\mu_n\}_{n=1}^\infty \) consists of probability measures supported on the compact set \( \gamma/2 \), and  is tight. By Prokhorov's theorem, there exists a subsequence \( \mu_{n_k} \) that converges weakly to a probability measure \( \tilde{\mu} \), supported on \( \gamma/2 \).

Since each \( \pi_{2n}(z) \) is a continuous function bounded on \( \gamma/2 \), the orthogonality relations pass to the limit:
\[
\int_{\gamma/2}  \pi_{2n}(z)\, \pi_{2m}(z) \, d\tilde{\mu}(z) = \delta_{n,m}, \qquad n,m \geq 0.
\]

Finally, define a symmetric measure \( \mu \) on \( \gamma \) by
\[
\mu(A) := \tfrac{1}{2} \left( \tilde{\mu}(A \cap \tfrac{1}{2}\gamma) + \tilde{\mu}(-A \cap \tfrac{1}{2}\gamma) \right),
\]
for Borel sets \( A \subset \gamma \). Then \( \mu \) is a symmetric probability measure on \( \gamma \), with respect to which the polynomials \( \pi_{2n}(z) \) are orthonormal, as claimed.
\end{proof}

\subsection{Heine like formula}\label{subsec:Heine-even}
{Heine formula \cite{Van_Assche_2020} for orthogonal polynomials in the real line allows us to view the polynomials as the expected characteristic polynomial of a joint probability distribution function (pdf). We prove analogous results in Proposition \ref{heine-formula} in the elliptic case when the weight function is symmetric.}

If we define 
\begin{equation}
    \nu_{k}=\int_{\gamma}\wp(z)^kw(z)dx, \qquad \widehat{\nu}_{k}=\frac{1}{4}\int_{\gamma}\wp'(z)^2\wp^{k}(z)w(z)dz
\end{equation}
with 

\begin{equation}
    \widehat{\nu}_k=4(\nu_{k+3}-g_2\nu_{k+1}-g_3\nu_{k})
\end{equation}
and the Hankel determinant of moments as 

\begin{equation}
    \Delta_{k}=\begin{vmatrix}
        \nu_0&\nu_1&\dots &\nu_{k-1}\\
        \nu_1&\nu_2&\dots& \nu_{k}\\
        \vdots&\vdots&\ddots&\vdots\\
        \nu_{k-1}&\nu_{k}&\dots & \nu_{2k-2}
    \end{vmatrix}\qquad     \widehat{\Delta}_{k}=\begin{vmatrix}
       \widehat{\nu}_0&\widehat{\nu}_1&\dots &\widehat{\nu}_{k-1}\\
        \widehat{\nu}_1&\widehat{\nu}_2&\dots&\widehat{\nu}_{k}\\
        \vdots&\vdots&\ddots&\vdots\\
        \widehat{\nu}_{k}&\widehat{\nu}_{k+1}&\dots &\widehat{\nu}_{2k-2}
    \end{vmatrix}.
\end{equation}

While $\Delta_0=\widehat{\Delta}_0=1$.

Then $P_{2k}$ is explicitly given as 

\begin{equation}\label{detex}
    P_{2k}(z)=\frac{1}{\Delta_{k}}\begin{vmatrix}
        \nu_0&\nu_1&\dots &\nu_{k}\\
        \nu_1&\nu_2&\dots& \nu_{k+1}\\
        \vdots&\vdots&\ddots&\vdots\\
        \nu_{k-1}&\nu_{k}&\dots & \nu_{2k-1}\\
        1&\wp(z)&\dots &\wp(z)^k
    \end{vmatrix}.
\end{equation}

While $P_{2k+3}$ is explicitly given as 

\begin{equation}\label{detex2}
    P_{2k+3}(z)=\frac{1}{\widehat{\Delta}_{k}}\begin{vmatrix}
        \widehat{\nu}_0&\widehat{\nu}_1&\dots &\widehat{\nu}_{k}\\
        \widehat{\nu}_1&\widehat{\nu}_2&\dots&\widehat{\nu}_{k+1}\\
        \vdots&\vdots&\ddots&\vdots\\
        \widehat{\nu}_{k+1}&\widehat{\nu}_{k+2}&\dots &\widehat{\nu}_{2k-1}\\
        -\frac{1}{2}\wp'(z)& -\frac{1}{2}\wp'(z)\wp(z)&\dots & -\frac{1}{2}\wp'(z)\wp(z)^k
    \end{vmatrix}
\end{equation}
\begin{proposition}\label{heine-formula}
We have
\begin{equation}\label{heine1}
    P_{2k}(z) = \frac{1}{k! \Delta_{k}} \int \dots \int \prod_{j=1}^{k} (\wp(z) - \wp(x_j)) \prod_{i<j} (\wp(x_j) - \wp(x_i))^2\, w(x_1) \dots w(x_k)\, dx_1 \dots dx_k,
\end{equation}
and
\begin{equation}\label{heine2}
\begin{aligned}
    P_{2k+3}(z) = \left( \frac{-1}{2} \right)^{2k+1} \frac{1}{k! \widehat{\Delta}_{k}} \int \dots \int \wp'(z) &\prod_{j=1}^{k} \wp'(x_j)^2 (\wp(z) - \wp(x_j)) \prod_{i<j} (\wp(x_j) - \wp(x_i))^2\,\\&\times  w(x_1) \dots w(x_k)\, dx_1 \dots dx_k.
\end{aligned}
\end{equation}
\end{proposition}

\begin{proof}
In \eqref{detex}, we can replace the moments by integrals, using \( x_j \) as the variable of integration in the \( j \)th row. Then, by multilinearity of the determinant, we extract \( k \) integrations out of the \( k \) rows to obtain
\begin{equation}\label{int1}
\begin{aligned}
   & P_{2k}(z) = \frac{1}{\Delta_{k}} \int \dots \int 
    \begin{vmatrix}
        1 & \wp(x_1) & \dots & \wp(x_1)^k \\
        \wp(x_2) & \wp(x_2)^2 & \dots & \wp(x_2)^{k+1} \\
        \vdots & \vdots & \ddots & \vdots \\
        \wp(x_k)^{k-1} & \wp(x_k)^k & \dots & \wp(x_k)^{2k-1} \\
        1 & \wp(z) & \dots & \wp(z)^k
    \end{vmatrix}
    w(x_1) \dots w(x_k)\, dx_1 \dots dx_k \\
    &= \frac{1}{\Delta_{k}} \int \dots \int 
    \prod_{j=1}^{k} \wp(x_j)^{j-1} 
    \prod_{i<j} \left( \wp(x_i) - \wp(x_j) \right)
    \prod_{j=1}^{k} \left( \wp(z) - \wp(x_j) \right)\,
    w(x_1) \dots w(x_k)\, dx_1 \dots dx_k.
\end{aligned}
\end{equation}

For a permutation \( \sigma \) of \( \{1,2,\dots,k\} \), we make a change of variables \( x_j \mapsto x_{\sigma(j)} \) in the integral \eqref{int1}. Then the Vandermonde factor \( \prod_{i<j} (\wp(x_i) - \wp(x_j)) \) changes sign if \( \sigma \) is odd, and we obtain
\begin{equation}
\begin{aligned}
    P_{2k}(z) = \frac{1}{\Delta_{k}} \int \dots \int 
    \mathrm{sgn}(\sigma) \prod_{j=1}^{k} \wp(x_{\sigma(j)})^{j-1}
    \prod_{i<j}& (\wp(x_i) - \wp(x_j))
    \prod_{j=1}^{k} (\wp(z) - \wp(x_j))\,
   \\& \times w(x_1) \dots w(x_k)\, dx_1 \dots dx_k.
\end{aligned}
\end{equation}
Averaging over all permutations, we get
\begin{equation}
\begin{aligned}
    P_{2k}(z) = \frac{1}{k! \Delta_{k}} \sum_{\sigma} \int \dots \int 
    \mathrm{sgn}(\sigma) \prod_{j=1}^{k} \wp(x_{\sigma(j)})^{j-1}
    \prod_{i<j}& (\wp(x_i) - \wp(x_j))
    \prod_{j=1}^{k} (\wp(z) - \wp(x_j))\,
    \\&\times w(x_1) \dots w(x_k)\, dx_1 \dots dx_k.
\end{aligned}
\end{equation}
Now observe that
\[
    \sum_{\sigma} \mathrm{sgn}(\sigma) \prod_{j=1}^{k} \wp(x_{\sigma(j)})^{j-1} 
    = \prod_{i<j} (\wp(x_i) - \wp(x_j)).
\]
Substituting this into the previous expression yields \eqref{heine1}. The formula \eqref{heine2} is obtained analogously.
\end{proof}

\subsection{Christoffel-Darboux formula for even EOPs}\label{subsec:CD-even}

We define the even Christoffel-Darboux Kernel as 
\begin{align}
&{K_{n}}(x,y)=\sqrt{w(x)}\sqrt{w(y)}\sum_{i=0}^{n-1}\pi_{2i}(x)\pi_{2i}(y)
\end{align}
The following Proposition and Lemma follow from standard theory such as \cite{forrester2010log}.

    The recurrence relation \eqref{3term2} gives the following expression
\begin{proposition}
    The kernel ${K_n}$ satisfies the Christoffel-Darboux identity
\begin{align}
    &K_{n}(x,y)=\alpha_n \sqrt{w(x)}\sqrt{w (y)}\frac{\pi_{2n}(x)\pi_{2n-2}(y)-\pi_{2n}(y)\pi_{2n-2}(x)}{\wp(x)-\wp(y)}.\\&
\end{align}
\end{proposition}

\begin{lemma}
    The determinant of the correlation kernel can be expressed as 
    \begin{equation}
          \frac{1}{n!}\det [K_n(x_i,x_j)]_{i,j=1}^{n}=\frac{1}{Z_{n}}\prod_{1\leq i<j\leq n}|\wp(x_{i})-\wp(x_{j})|^{2}\prod_{i=1}^{n}w(x_{i}),
    \end{equation}
    where the partition function $Z_n$ is given by 
    \begin{equation}
        Z_{n}=\int_{\gamma}\dots\int_{\gamma}\prod_{1\leq i<j\leq n}|\wp(x_{i})-\wp(x_{j})|^{2}\prod_{i=1}^{n}w(x_{i})dx_{1}\dots dx_{n}.
    \end{equation}
\end{lemma}

\begin{remark}[A curious identity]
We assume $w(x)\equiv 1$. In that case we note $\pi_{n}'(x)$ is a polynomial of degreen $n+1$. If we expand in the orthonormal polynomial basis we can write 
\begin{equation}
    \pi_n'(x)=\sum_{j=0}^{n+1}c_{j,n}\pi_{j}(x),\qquad c_{j,n}=\int_\gamma \pi_{n}'(x)\pi_{j}(x)dx.
\end{equation}
Then integration by parts, and taking into account the double periodicity of the polynomials, we find that $c_{j,n}=0$ for $j\leq n-2$. As a corollary we then obtain 
\begin{equation}
    \pi_n'(x)=\sum_{j=n-1}^{n+1}c_{j,n}\pi_{j}(x).
\end{equation}
\end{remark}
\subsection*{Acknowledgements}
We thank Peter Forrester and Arno Kuijlaars for their insights and meaningful discussions.
H.D acknowledges the generous support of the SMRI Postdoctoral Fellowship.
S.L. acknowledges financial support from the International Research Training Group (IRTG) between KU Leuven and University of Melbourne and Melbourne Research Scholarship of University of Melbourne.
\appendix
\section{Simplicity of zeros}\label{simplezero}
Let us first start by proving the analogue of \cite[Theorem 2.3, 2.4]{bertola2022nonlinear} for our EOPs \footnote{The number of zeros and the proof of this statement are modified in our case as the definition of polynomials in our case differ from that of \cite{bertola2022nonlinear}.}.
\begin{theorem}\label{zerothm}
    Let \( n \) be even. Then \( \pi_n \) has exactly \( n \) zeros on \( \gamma \). If \( n \) is odd, \( \pi_n \) has \( n-1 \)  zeros on \( \gamma \), and the remaining zero lies on the interval $[0,1]$.
\end{theorem}

\begin{proof}
   We begin with the fact that the polynomial $\pi_n(z)$ has $n$ zeros, as it an elliptic function with an $n^{th}$ order pole. Furthermore, by Abel's theorem \cite{miranda}, the sum of the divisors of the zeros must be zero, and since the only pole of EOPs is at $z=0$, the sum of all zeros of $\pi_n(z)$ must be zero modulo the lattice. 
   
   Let us now denote the zeros of \( \pi_n \) on \( \gamma \) by \( s_1, \dots, s_k \), the zeros on the interval $[0,1]$ by \( t_1, \dots, t_l \), and the zeros elsewhere by \( x_1, \dots, x_j \), where $k+l+j =n$.  Since the polynomial \( \pi_n \) has real coefficients, for each zero \( x_i \) on the torus, its complex conjugate \( \overline{x_i} \) is also a zero. This symmetry implies that the sum of the imaginary parts of the \( x_i \)'s is zero. Similarly, the sum of the imaginary parts of the zeros \( t_i \) on the real line is trivially zero. Finally, since the imaginary part of \( \gamma \) is \( \Im(\gamma) = \frac{\tau}{2} \), the number of zeros on \( \gamma \) must be even.

    Let \( k=2m \leq n-2 \) denote the number of zeros on \( \gamma \). In this case, we can construct an elliptic polynomial \( P_{2m+1} \) with zeros at \( s_1, \dots, s_{2m} \). By Abel’s theorem, the location of the remaining zero is uniquely determined, and it must lie on the real line. Since \( P_{2m+1} \) is a polynomial of degree less than \( n-1 \), its zeros on $\gamma$ match those of \( \pi_n \) with the same multiplicity.

    However, this leads to a contradiction. Specifically, we have
    \[
    \int_{\gamma} \pi_n(z) P_{2m+1}(z) \, w(z)\, dz > 0,
    \]
    which contradicts the orthogonality condition.
    Now, suppose \( n \) is even. In this case, the number of zeros on \( \gamma \) must be either \( n \) or \( n-1 \). However, the number of zeros must be even, so it cannot be \( n-1 \). Thus, \( \pi_n \) must have exactly \( n \) zeros on \( \gamma \). Finally, suppose \( n \) is odd. By the same reasoning, \( \pi_n \) must have exactly \( n-1 \) zeros on \( \gamma \). Abel's theorem then implies that the remaining zero must lie on the real line.
\end{proof}

\begin{corollary}
    The zeros of \( \pi_n \) are all simple.
\end{corollary}
\begin{proof}
    Suppose, for contradiction, that \( \pi_n \) has a zero of multiplicity greater than one. Let \( z_1, \dots, z_k \) be the distinct zeros of \( \pi_n \) that have odd multiplicity strictly greater than one. Since the total number of such zeros must be even, it follows that \( k \) is even.

    Consider first the case where \( n \) is odd. From Theorem \ref{zerothm}, we know that \( \pi_n \) has exactly \( n-1 \)  zeros on \( \gamma \) and a single zero on the real line. In particular, if not all zeros are simple, then \( k \leq n - 3 \). 

    We now construct an elliptic polynomial  \( P_k \), on \( \mathbb{T} \), with simple zeros at \( z_1, \dots, z_k \). By Abel’s theorem, the location of the remaining pole is uniquely determined so as to ensure that the divisor of \( P_k \) sums to zero. Since the divisor of \( P_k \) includes exactly \( k + 1 \) simple zeros, \( P_k \) is of degree at most \( k + 1 \leq n - 2 \).

    Since \( \pi_n \) and \( P_k \) vanish at the same points (modulo multiplicity), their product does not change sign on \( \gamma \), and in particular,
    \[
        \int_\gamma \pi_n(z) P_k(z) \, w(z)\, dz > 0,
    \]
    which contradicts the orthogonality of \( \pi_n \) with all functions of degree less than \( n \).

    The case where \( n \) is even is treated analogously: any non-simple zero yields a construction of \( P_k \) of degree at most \( n - 2 \), leading to the same contradiction.

    Hence, all zeros of \( \pi_n \) must be simple.
\end{proof}
\section{Relation between recurrence coefficients}\label{App:rel_rec_coeff}

Our goal is to relate the recurrence coefficients $\{a_n\},\{b_n\}$ and $\{c_n\}$ in the 5 term recurrence relation with the recurrence coeffcients $\{p_n\},\{q_n\},\{r_n\},\{s_n\}$ from the 7 term recurrence relation. We begin with the Weierstrass elliptic curve equation:
\begin{equation} \label{eq:elliptic_curve}
\wp'(z)^2 = \wp(z)^3 - g_2 \wp(z) - g_3.
\end{equation}
This identity allows one to express \( \wp'(z)^2 \pi_n(z) \) in terms of \( (\wp(z)^3 - g_2 \wp(z) - g_3)\pi_n(z) \), which, when expanded via recurrence relations, yields polynomial combinations of \( \pi_{n+k}(z) \) for \( -6 \le k \le 6 \). Matching coefficients of like terms gives algebraic identities between recurrence coefficients.

More precisely we can obtain 
\begin{equation}
   \left( \wp(z)^3-g_{2}\wp(z)-g_3\right)\pi_{n}(z)=\sum_{i=-6}^{6}B_{n+i}\pi_{n+i}.
\end{equation}
by using the five-term recurrence relation thrice. The computation is tedious and done by computer software. 

\begin{align*}
B_{n-6} &= 4 a_{n-5} a_{n-3} a_{n-1}, \\
B_{n-5} &= 4 (a_{n-3} b_{n+2} + a_{n+1} b_{n-2}) a_{n-4} + 4 a_{n-3} a_{n+1} b_{n-4}, \\
B_{n-4} &= 4 (a_{n-3} b_{n+2} + a_{n+1} b_{n-2}) b_{n-3} + 4 (a_{n-3} c_{n+2} + a_{n+1} c_{n-2} + b_{n-2} b_{n+2}) a_{n-3} \\&\quad + 4 a_{n-3} a_{n+1} c_{n-4}, \\
B_{n-3} &= 4 (a_{n-3} b_{n+2} + a_{n+1} b_{n-2}) c_{n-3} + 4 (a_{n-3} c_{n+2} + a_{n+1} c_{n-2} + b_{n-2} b_{n+2}) b_{n-2} \\&\quad + 4 (a_{n+1} b_{n-1} + b_n c_{n+2} + b_{n+2} c_{n-1}) a_{n-2} + 4 a_{n-3} a_{n+1} b_{n-3}, \\
B_{n-2} &= -g_2 a_{n-1} + 4 (a_{n-3} b_{n+2} + a_{n+1} b_{n-2}) b_{n-2} + 4 (a_{n-3} c_{n+2} + a_{n+1} c_{n-2} + b_{n-2} b_{n+2}) c_{n-2} \\&\quad + 4 (a_{n-1} a_{n+1} + b_{n-1} b_{n+2} + c_n c_{n+2}) a_{n-1} + 4 (a_{n+1} b_{n-1} + b_n c_{n+2} + b_{n+2} c_{n-1}) b_{n-1} \\&\quad + 4 a_{n-3}^2 a_{n+1}, \\
B_{n-1} &= -g_2 b_n + 4 (a_{n-3} b_{n+2} + a_{n+1} b_{n-2}) a_{n-2} + 4 (a_{n-3} c_{n+2} + a_{n+1} c_{n-2} + b_{n-2} b_{n+2}) b_{n-1} \\&\quad + 4 (a_{n-1} a_{n+1} + b_{n-1} b_{n+2} + c_n c_{n+2}) b_n + 4 (a_{n-1} b_{n+2} + b_{n+1} c_{n+2} + b_{n+3} c_n) a_n \\&\quad + 4 (a_{n+1} b_{n-1} + b_n c_{n+2} + b_{n+2} c_{n-1}) c_{n-1}, \\
B_n &= -g_2 c_n - g_3 + 4 (a_{n-3} c_{n+2} + a_{n+1} c_{n-2} + b_{n-2} b_{n+2}) a_{n-1} \\&\quad + 4 (a_{n-1} a_{n+1} + b_{n-1} b_{n+2} + c_n c_{n+2}) c_n + 4 (a_{n-1} b_{n+2} + b_{n+1} c_{n+2} + b_{n+3} c_n) b_{n+1} \\&\quad + 4 (a_{n+1} b_{n-1} + b_n c_{n+2} + b_{n+2} c_{n-1}) b_n + 4 (a_{n+1} c_{n+2} + a_{n+3} c_n + b_{n+1} b_{n+3}) a_{n+1}, \\
B_{n+1} &= -g_2 b_{n+1} + 4 (a_{n+1} b_{n+3} + a_{n+3} b_{n+1}) a_{n+2} + 4 (a_{n-1} a_{n+1} + b_{n-1} b_{n+2} + c_n c_{n+2}) b_{n+1} \\&\quad + 4 (a_{n-1} b_{n+2} + b_{n+1} c_{n+2} + b_{n+3} c_n) c_{n+1} + 4 (a_{n+1} b_{n-1} + b_n c_{n+2} + b_{n+2} c_{n-1}) a_n \\&\quad + 4 (a_{n+1} c_{n+2} + a_{n+3} c_n + b_{n+1} b_{n+3}) b_{n+2}, \\
B_{n+2} &= -g_2 a_{n+1} + 4 (a_{n+1} b_{n+3} + a_{n+3} b_{n+1}) b_{n+3} + 4 (a_{n-1} a_{n+1} + b_{n-1} b_{n+2} + c_n c_{n+2}) a_{n+1} \\&\quad + 4 (a_{n-1} b_{n+2} + b_{n+1} c_{n+2} + b_{n+3} c_n) b_{n+2} + 4 (a_{n+1} c_{n+2} + a_{n+3} c_n + b_{n+1} b_{n+3}) c_{n+2} \\&\quad + 4 a_{n+1} a_{n+3}^2, \\
B_{n+3} &= 4 (a_{n+1} b_{n+3} + a_{n+3} b_{n+1}) c_{n+3} + 4 (a_{n-1} b_{n+2} + b_{n+1} c_{n+2} + b_{n+3} c_n) a_{n+2} \\&\quad + 4 (a_{n+1} c_{n+2} + a_{n+3} c_n + b_{n+1} b_{n+3}) b_{n+3} + 4 a_{n+1} a_{n+3} b_{n+4}, \\
B_{n+4} &= 4 (a_{n+1} b_{n+3} + a_{n+3} b_{n+1}) b_{n+4} + 4 (a_{n+1} c_{n+2} + a_{n+3} c_n + b_{n+1} b_{n+3}) a_{n+3} + 4 a_{n+1} a_{n+3} c_{n+4}, \\
B_{n+5} &= 4 (a_{n+1} b_{n+3} + a_{n+3} b_{n+1}) a_{n+4} + 4 a_{n+1} a_{n+3} b_{n+5}, \\
B_{n+6} &= 4 a_{n+1} a_{n+3} a_{n+5}.
\end{align*}
Similarly we compute 

\begin{equation}
    \wp'(z)^2\pi_{n}(z)=\sum_{i=-6}^{6}B_{n+i}\pi_{n+i}
\end{equation}
using the 7 term recurrence relation twice.
\begin{align*}
B_{n+6} &= p_{n+3} p_{n+6}, \\
B_{n+5} &= p_{n+3} q_{n+5} + q_{n+2} p_{n+5}, \\
B_{n+4} &= p_{n+3} r_{n+4} + q_{n+2} q_{n+4} + r_{n+1} p_{n+4}, \\
B_{n+3} &= p_{n+3} s_{n+3} + q_{n+2} r_{n+3} + r_{n+1} q_{n+3} + s_n p_{n+3}, \\
B_{n+2} &= q_{n+2} s_{n+2} + r_{n+1} r_{n+2} + s_n q_{n+2} + r_n p_{n+2}+p_{n+3}r_{n+3}, \\
B_{n+1} &= r_{n+1} s_{n+1} + s_n r_{n+1} + r_n q_{n+1} + q_n p_{n+1}+q_{n+2}r_{n+2}+p_{n+3}q_{n+3}, \\
B_n &= p_{n+3}^2 + q_{n+2}^2 + r_{n+1}^2 + s_n^2 + r_n^2 + q_n^2 + p_n^2, \\
B_{n-1} &= s_n r_n + r_n s_{n-1} + q_n q_{n+1} + p_n r_{n-1}+p_{n}q_{n-1}+q_{n}r_{n-1}, \\
B_{n-2} &= s_n q_n + r_n r_{n-1} + q_n s_{n-2} + p_n q_{n-1}, \\
B_{n-3} &= s_n p_n + r_n q_{n-1} + q_n r_{n-2} + p_n s_{n-3}, \\
B_{n-4} &= r_n p_{n-1} + q_n q_{n-2} + p_n r_{n-3}, \\
B_{n-5} &= q_n p_{n-2} + p_n q_{n-3}, \\
B_{n-6} &= p_{n} p_{n-3}.
\end{align*}

\bibliographystyle{plainnat} 
\bibliography{biblio}

\end{document}